\documentclass[conference, 10pt, letterpaper]{ieeeconf}
\IEEEoverridecommandlockouts
\pdfoutput=1

\makeatletter

\let\proof\@undefined
\let\endproof\@undefined
\makeatother

\usepackage{algorithm}
\usepackage{amsmath}
\usepackage{amssymb}
\usepackage{amsthm}

\usepackage[font=small]{caption}
\usepackage[font=footnotesize]{subcaption}
\usepackage{booktabs}
\usepackage{multirow}
\newtheorem{theorem}{Theorem}[section]

\newtheorem{lemma}[theorem]{Lemma}

\theoremstyle{definition}

\newtheorem{problem}[theorem]{Problem}

\theoremstyle{remark}
\newtheorem{remark}[theorem]{Remark}

\theoremstyle{definition}

\usepackage{tikz}
\usetikzlibrary{automata,positioning}
\usepackage{tkz-berge}

\usepackage{epstopdf}
\usepackage{graphicx,url}
\usepackage{epstopdf}
\usepackage{color}
\usepackage{float}
\usepackage{comment}
\usepackage{mathtools}

\DeclareMathOperator*{\argmax}{arg\,max}

\renewcommand{\thesection}{\Roman{section}}
\renewcommand{\thesubsection}{\thesection.\Roman{subsection}}
\usepackage{algpseudocode}

\newcommand\oprocendsymbol{\hbox{$\bullet$}}
\newcommand\oprocend{\relax\ifmmode\else\unskip\hfill\fi\oprocendsymbol}

\urldef{\smith}\url{stephen.smith@uwaterloo.ca}
\urldef{\armin}\url{a6sadegh@uwaterloo.ca}

\renewcommand{\baselinestretch}{0.985}

\title{\Large \textbf{On Re-Balancing Self-Interested Agents in Ride-Sourcing Transportation Networks}}
\author{Armin Sadeghi \qquad Stephen L. Smith \thanks{This research is partially
    supported by the Natural Sciences and Engineering Research Council
    of Canada (NSERC). }
  \thanks{The authors are with the Department of Electrical and
    Computer Engineering, University of Waterloo, Waterloo ON, N2L 3G1    Canada  (\armin; \smith) }}

\begin{document}
\maketitle

\begin{abstract}
This paper focuses on the problem of controlling self-interested drivers in ride-sourcing applications. Each driver has the objective of maximizing its profit, while the ride-sourcing company focuses on customer experience by seeking to minimizing the expected wait time for pick-up. These objectives are not usually aligned, and the company has no direct control on the waiting locations of the drivers. In this paper, we provide two indirect control methods to optimize the set of waiting locations of the drivers, thereby minimizing the expected wait time of the customers: 1) sharing the location of all drivers with a subset of drivers, and 2) paying the drivers to relocate. We show that finding the optimal control for each method is NP-hard and we provide algorithms to find near-optimal control in each case. We evaluate the performance of the proposed control methods on real-world data and show that we can achieve between $20\%$ to $80\%$ improvement in the expected response.
\end{abstract}

%
%
\section{Introduction}
\label{sec:introduction}

In recent years, ride-sourcing services such as UberX and Lyft have emerged as an alternative mode of urban transportation. The reduced wait times for pickup of these services is the compelling feature compared to the conventional taxi services~\cite{rayle2014app}. A key factor that affects the response time of the service is where the drivers wait to respond to the next ride request.  Ride-sourcing companies do not have control over the position of the drivers as they are self-interested units maximizing their objectives. Therefore, a challenge is to ensure the drivers are distributed throughout the city in order to minimize the expected wait time of the customers.  This must be done either by providing information to drivers, or through incentives (payment) that make relocation attractive.  

As a method for both re-balancing and increasing the supply of drivers, Uber introduced surge pricing in high demand areas.  This reduces the expected response time of servicing the requests by drawing more drivers to those areas. However, the surge pricing can draw drivers away from lower demand areas, resulting in higher wait times in those areas and more imbalance~\cite{surgeprice, rosenblat2016algorithmic}.

The problem of servicing requests in ride-sourcing networks can be divided into two major problems: 1) assignment of the ride request to the drivers; and 2) re-balancing of the drivers for future ride requests. In this paper, we focus on the re-balancing problem for a subset of drivers to service ride requests that arrive sequentially in an environment. The drivers' motion in the environment is captured as a road-map (i.e., graph), and each ride request arrives at a node of the graph according to a known arrival rate (see Figure~\ref{fig:intro_fig}). The drivers are self-interested units maximizing the expected profit of their workday. Hence, the ride-sourcing company has no direct control over the waiting locations of the drivers. The objective of the ride-sourcing company is to incentivize the drivers to relocate to a set of waiting locations that minimizes the expected wait-time of the customers.

\emph{Related Work:}  The problem of dispatching taxis to service ride requests arriving sequentially over time has been the subject of extensive research~\cite{zhang2015performance, hyland2018dynamic, DBLP:journals/corr/abs-1805-02014, maciejewski2016assignment}. These studies focus on policies to optimally assign the ride requests to the taxis. In contrast, we focus on the waiting locations of the drivers that minimize the expected wait time of the customers. We assume the ride requests are assigned in a first-come-first-serve fashion to the closest available driver. Assigning the closest available driver to each request is the common method employed by the ride-sourcing companies~\cite{uberassign}.

\begin{figure}
\centering
    \includegraphics[width=\linewidth]{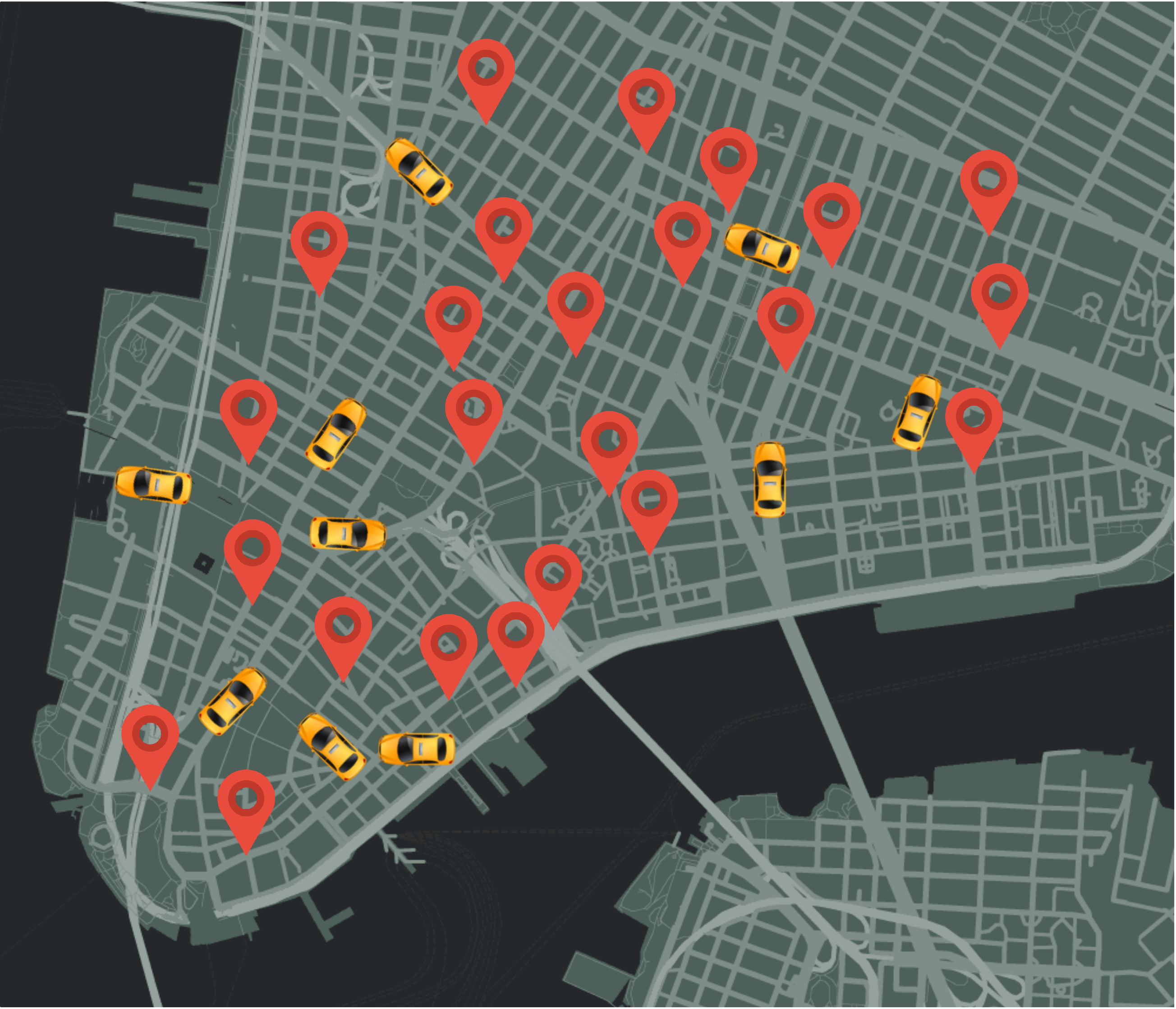}
     \caption{A set of drivers in the ride-sourcing system and a set of locations with high probability of ride request arrival.}
\label{fig:intro_fig}
\end{figure}

The problem of re-balancing service units in the environment has been studied for various applications. In mobility-on-demand problem (MOD)~\cite{pavone2012robotic, tsao2018stochastic, calafiore2017flow}, a group of vehicles are located at a set of stations. The customers arrive at the stations, hire the vehicles for ride, and then drop the vehicles off at their destination stations. The objective is to balance the vehicles at the stations to minimize the expected wait time of the customer. In comparison to MOD, we consider the customer wait-time as the time between the request arrival and the pick-up time, which incorporates the distance of the closest available vehicle to the pick-up location.

The facility location problem~\cite{shmoys2000approximation, arya2004local} and its extension to the mobile
facility location problem (MFL)~\cite{demaine2009minimizing} is the problem of
distributing facilities in a set of locations to respond to the demands arriving at different locations. The objective is to minimize the time to respond to the
demands and the total cost of opening facilities. A special case
of the facility location problem is the $k$-median problem~\cite{arya2004local} where the number of open facilities is limited and the cost of opening a facility is zero. In~\cite{sadeghi2018re}, we addressed a multi-stage MFL problem where we relocate a set of autonomous vehicles to minimize expected response time for future requests in a receding horizon manner.


In the aforementioned studies, the actions of service units are controlled by a central unit and the objective of the service units is aligned with the global objective. However, we consider self-interested units maximizing their own profit such that the ride-sourcing companies have no direct control on their decisions. A closely related problem is that of Voronoi games on graphs~\cite{bandyapadhyay2015voronoi} where requests arrive on the vertices of a graph and the objective of each self-interested service unit is to maximize the number vertices assigned to them. They cast the problem as a game between the service units prove that the problem of finding the pure Nash equilibrium on general graphs is NP-hard. In~\cite{salhab2017dynamic}, the authors provide the best response strategy for each driver and they approximate Nash equilibria, which can be utilized by the ride-sourcing companies to incentivize the drivers in a way to maximizes a global objective. These studies focus on the strategies of the self-interested service units, in contrast, we focus on finding the optimal policy for the ride-sourcing company to optimally respond to the ride requests.

\emph{Contributions:}
The contributions of this paper are threefold. First, we formulate the re-balancing problem of self-interested service units. Second, we propose two indirect method to relocate the service units to minimize the expected response time and provide algorithms with near-optimal solutions. Third and finally, we evaluate the performance of the two proposed control methods on real-world ride-sourcing data.

The paper is organized as follows. In Section~\ref{sec:formulation}, we
formulate the problem of minimizing customer wait-time with self-interested service units. In Section~\ref{sec:sharing_info},
we provide the first indirect control method based on sharing the information on the location of the drivers. Section~\ref{sec:service_providers-game}
consists of the second control method based on incentive pay for the drivers to relocate to desired waiting locations. Finally in
Section~\ref{sec:simulation_results}, we provide an extensive set of experiments, on real-world ride-sourcing data, characterizing the performance of the two control methods.

%
%

\section{Problem Formulation}
\label{sec:formulation}
Consider a set of $m$ drivers and a set of pick-up and drop-off locations $V$. Let $G=(V, E, c)$ be a metric graph on vertices $V$, let $E$ be the set of edges between the locations and $c: E \rightarrow \mathbb{R}_+$ be the function assigning a travel time to each edge of the graph. The drivers wait on a subset of the vertices for the next request, which we call the configuration of the drivers $Q$. The set of all configurations of the drivers is denoted by $\mathcal{Q}$. Each driver $i$ is aware of the position of a subset of the drivers $I_i \subseteq Q$,  for instance, each driver may be aware of the location of the other drivers in its vicinity.

We assume that the requests arrive at each vertex $u$ according to an independent Poisson process with arrival rate $\lambda_u$.  Upon a request arrival, the closest driver to the vertex of the request is assigned to service the request. Let $p_a(u)$ denote the arrival probability, which is the ratio of number of requests arriving at $u$ to the total number of requests arriving in a period of time.  Let the drop-off probability $p_d(\mathrm{dropoff} = w|\mathrm{pickup} =v)$ be the probability that a request with pickup location at vertex $v$ has a drop-off location at vertex $w$.

Driver $i$'s \emph{perception of her expected profit} is a function of her information $I_i$ on the location of other drivers, environment parameters such as arrival times, her waiting location $q_i$ and the period of working time $B_i$, denoted by $\mathcal{V}_i(u, B_i)$.  For the development of our main control methods we do not assume any specific form of this function.  We do assume, however, that the ride-sourcing company has access to this function, obtained through data of driver behavior.  In Section~\ref{sec:simulation_results} we present one potential model of $\mathcal{V}_i(u, B_i)$, which is then used for simulating the two control methods.


\emph{Drivers' objective:} Each driver is a self-interested unit, therefore, they will wait at a location that maximizes their expected profit, i.e.,
\begin{align}
    \label{eq:driver_utility}
    \argmax_{u \in V} -\sigma c(q_i, u) &+ \mathcal{V}_i(u, B_i - c(q_i, u)),    
\end{align}
where $\sigma$ is the cost per minute of driving. 

\emph{Global objective:} In addition to the objective of each driver, there is a global objective for the service providers such as Uber and Lyft to maximize the service quality by minimizing the expected wait time of the customers until pick-up, i.e., 
\begin{equation}
\label{eq:global_obj}
    \min _{Q \in \mathcal{Q}} D(Q) = \sum_{u\in V}\min_{q_i \in Q} p_a(u) c(q_i, u),  
\end{equation}

The main challenge in optimizing the global objective is that the drivers are self-interested units and the service provider does not have any direct control on the configuration of the vehicles. Therefore, the service provider is not able to minimize the expected response time to the requests directly. The two indirect control methods proposed in this paper incentivize the drivers to relocate to desired waiting locations. The first control method exploits the dependency of the expected profit of the drivers on their information $I_i$. The service provider can share more information on the location of drivers with a subset of them to manipulate their decision towards relocating to a desired waiting location.  We refer to this as the \emph{sharing information} control method. 
The second proposed control method, incentivizes the drivers to relocate to desired waiting locations with payments, which we refer to as the \emph{pay-to-control} method. These control methods are applicable to various models of driver behaviour $\mathcal{V}$.  

 
In the following sections, we provide a detailed description of the two control methods and propose algorithms to find near optimal controls.

\section{Control by Sharing Information}
\label{sec:sharing_info}
 In this section, we provide an indirect control on the configuration of the drivers exploiting the fact that the optimal waiting location for each driver in Equation~\eqref{eq:driver_utility} is a function of the information provided to the driver regarding the position of the other drivers, i.e., $I_i$.

Figure~\ref{fig:info_share_example} demonstrates the importance of information on an instance of the ride-sourcing problem with two vehicles and two request locations. The locations are within unit distance apart and the arrival rate at locations $v_1$ and $v_2$ are $0.1$ and $0.2$, respectively.  The vehicles are initially located at $v_1$ and will relocate to the best waiting location, namely optimizing Equation~\eqref{eq:driver_utility}. Figure~\ref{subfig:info_share_fair} shows the two scenarios where both vehicles are provided the same information, i.e., $I_1=I_2= \emptyset$ and $I_1=I_2= Q$. Note that the configuration of the vehicles when they are provided the same information is the worst possible configuration for the global objective. However, illustrated in Figure~\ref{subfig:info_share_partial}, providing the information to a subset of the vehicles results in the optimal configuration for the global objective.
  \begin{figure}
      \centering
      \begin{subfigure}[t]{0.49\linewidth}
      \includegraphics[width=\textwidth]{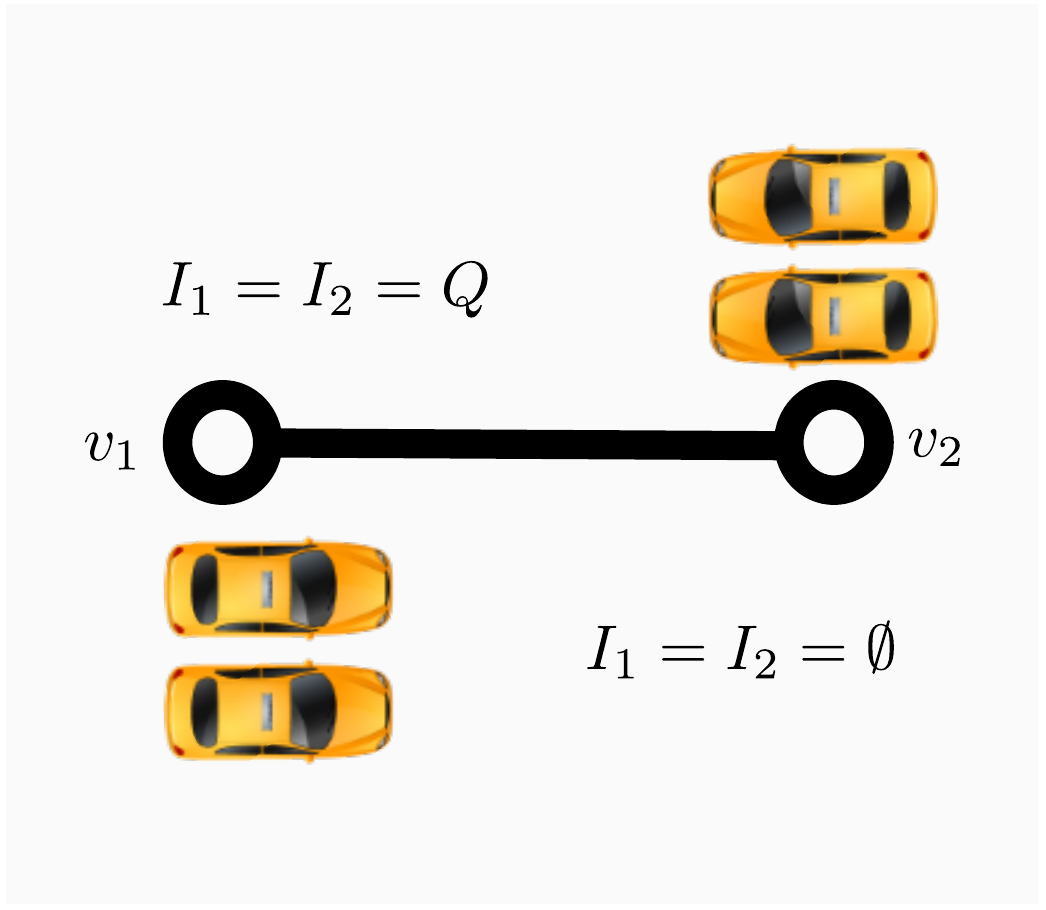}
      \caption{Equal information sharing}
      \label{subfig:info_share_fair}
      \end{subfigure}
      \begin{subfigure}[t]{0.49\linewidth}
      \includegraphics[width=\textwidth]{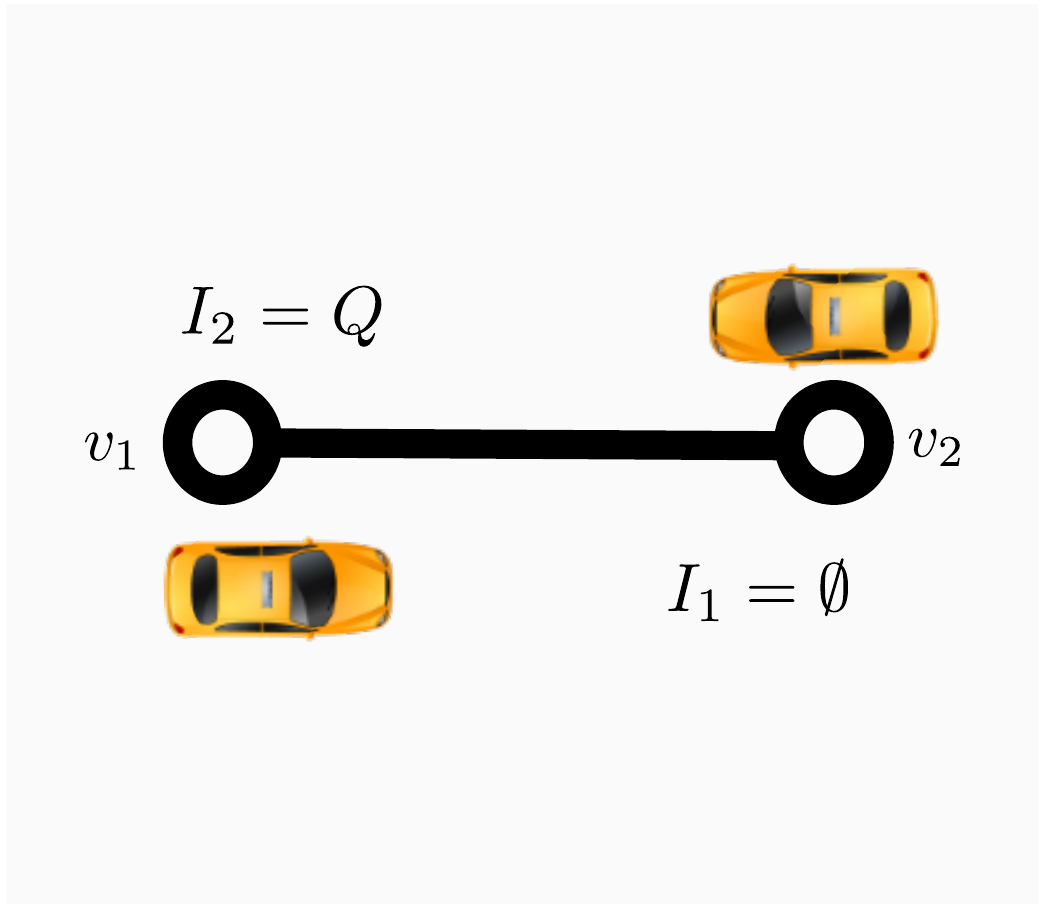}
      \caption{Partial information sharing }
      \label{subfig:info_share_partial}
      \end{subfigure}
      \caption{Instance of ride-sourcing problem with two vehicles and two request arrival locations. }
      \label{fig:info_share_example}
  \end{figure}
  
    The information sharing problem consists of deciding the subset of drivers we share  information with and the information shared with each driver. However, for this work, we consider the binary decision where either full information or no information is provided.
  
  Let $q'_{i, Q}$ (resp. $q'_{i, \emptyset}$) be the new waiting location selected by driver $i$ from Equation~\eqref{eq:driver_utility} with information $I_i = Q$ (resp. $I_i = \emptyset$). Let $F_i = \{q'_{i,Q}, q'_{i,\emptyset}\}$ be the set of candidate waiting locations for driver $i$. If there exist $i, j \in [m]$ such that $q'_{i, I_i} = q'_{j, I_j}$, then we create a duplicate vertex $u$ for $q'_{j, I_j}$ such that $c(q'_{i, I_j}, v) = c(u, v)$ for all $v \in V$. 
  The formal definition of the problem of sharing information with drivers is given as follows:
\begin{problem}
\label{prob:share_info}
Consider a metric graph $G = (\cup_{i = 1}^{m}F_i\cup V, E, c)$. Find a new configuration $Q'$ by picking only one vertex from each $F_i$, i.e., such that $|Q' \cap F_i| = 1$ for each $i$, while minimizing  the global objective $D(Q') = \sum_{u\in V}\min_{q_i \in Q'} p_a(u) c(q, u)$.
\end{problem} 

\begin{figure}
\centering
    \includegraphics[width=\linewidth]{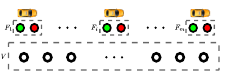}
     \caption{An instance of Problem~\ref{prob:share_info}. The green vertices represent the desired waiting location of each driver if $I_i = \emptyset$, and the red vertices represent the waiting location of the drivers if $I_i = Q$.}
\label{fig:information_sharing_problem}
\end{figure}
Figure~\ref{fig:information_sharing_problem} shows an instance of the information sharing problem. The green vertices are the waiting locations of the drivers if they have no information on the position of the other drivers, and the red vertices are the waiting locations of the drivers if the information is provided to each driver by the service provider. Let $Q'$ be the solution to Problem~\ref{prob:share_info} in which if $q_{i, Q} \in Q'$ then the driver $i$ is provided complete information of the position of the other drivers, and no information is available for driver $i$ if $q'_{i, \emptyset} \in Q'$.    

First, we analyze the complexity of Problem~\ref{prob:share_info}, and then we provide an LP-rounding algorithm to find the controls.

\subsection{Proof of Hardness}
We now prove the NP-hardness of Problem~\ref{prob:share_info} with a reduction from CNF-SAT~\cite{schuler2005algorithm} as follows:

Consider  an  instance  of  CNF-SAT  with $n$ Boolean  variables  and $m$ clauses.  We will reduce this problem to Problem~\ref{prob:share_info}.
\begin{enumerate}
\item Let $\cup_{i=1}^{m}F_i$ contain $2n$ vertices, partitioned into $n$ sets of size two. The set $F_i$ contains two vertices, $\{v_i^T, v_i^F\}$, where $v_i^T$ will correspond to setting the $i$th SAT variable to true (i.e., the positive literal) and $v_i^F$ will correspond to setting it to false (i.e., the negative literal). 
\item We let $V$ contain $m$ vertices, one representing each clause in the SAT formula.
\item  Let $E$ contain an edge for each $v \in \cup_{i=1}^{m}F_i$ and $w \in V$.
\item For each $e = (v, w) \in E$, we set its cost to $1$ if the literal $v$ appears in the clause $w$, and $2$ if the literal does not. Note that the costs are metric.
\item Let $P_u = 1/m$ for each $u \in V$.
\end{enumerate}

Now, we solve the instance of Problem~\ref{prob:share_info}. If it returns a subset of $\cup_{i=1}^{m}F_i$ with cost exactly $1$, then for each clause $c \in V$, there is literal in $\cup_{i=1}^{m}F_i$ with edge cost of $1$ to $c$.  This implies that the literal chosen from each subset in the partition of $\cup_{i=1}^{m}F_i$ gives a satisfying truth assignment for the SAT instance.  If the subset returned has a cost greater than $1$, then there exists a clause $w \in V$ for which every chosen literal has edge cost of $2$. Thus, this clause is not satisfied and no satisfying instance exists.

\subsection{Linear-Program Rounding Algorithm}
Given that Problem~\ref{prob:share_info} is NP-hard, we turn our focus to suboptimal algorithms.  In particular, we provide a simple Linear Program (LP)-rounding algorithm for Problem~\ref{prob:share_info}. Although we do not provide bound on the performance of the LP-rounding algorithm, we evaluate the performance of the algorithm on an extensive set of real-world ride sourcing data in Section~\ref{sec:simulation_results} and we show that the proposed algorithm is on average within $0.014\%$ of the optimal.

First we cast Problem~\ref{prob:share_info} as an integer linear program (ILP), then we propose a rounding algorithm based on the solution to the relaxation of the ILP. Let integer parameter $x_{u, v} \in \{0, 1\}$ denote the assignment of a request at $v$ to a driver at vertex $u$ if $x_{u, v} = 1$, and $x_{u, v} = 0$ otherwise. Let the integer parameter $y_u \in \{0, 1\}$ for all $u \in \cup_{i}^{m}F_i$ represent if there is a driver assigned to wait for next request arrival at $u$. Then we write the ILP for Problem~\ref{prob:share_info} as follows:

\begin{align}
\label{eq:ILP_share_info}
\text{minimize}&\quad \sum_{v \in V}^{}\sum_{u \in \cup_{i}^{m}F_i}^{} p_a(v)c(u, v)x_{u,v }\\
\text{subject to:}& \sum\limits_{u \in \cup_{i}^{m}F_i}   x_{u, v} \geq 1,  \quad \ \forall v \in V\label{eq:constraint_service_x}\\
    &y_u \geq x_{u, v},  \quad \ \forall v \in V, u \in \cup_{i}^{m}   F_i\label{eq:constraint_y_g_x} \\
    &y_u + y_v  = 1,  \quad \ F_i = \{u, v\}, i \in [m]\label{eq:constraint_one_in_set} \\
    &y_u, x_{u,v} \in \{0,1\}, \quad \ \forall v \in V, u \cup_{i}^{m}F_i\nonumber
\end{align}
By constraint~\eqref{eq:constraint_service_x}, a feasible solution assigns each request location to a driver. Equation~\eqref{eq:constraint_y_g_x} ensures that a request is assigned to $u$ only if there is a driver located at $u$, and finally Equation~\eqref{eq:constraint_one_in_set} shows that in a feasible solution only one of the candidate waiting locations is chosen from each subset $F_i$, which represent that either the information is provided to a driver or otherwise.

Now we propose our LP-rounding algorithm for Problem~\ref{prob:share_info}. Let $(\mathbf{x}', \mathbf{y}')$ be the solution to the LP relaxation of ILP~\eqref{eq:ILP_share_info}. Without loss of generality for all $F_i = \{u, v\}, i \in [m]$, let $y_u \geq 1/2$ and $y_v \leq 1/2$. Given solution $(\mathbf{x}', \mathbf{y}')$ we construct an integer solution to ILP~\eqref{eq:ILP_share_info} by setting $y_u = 1$ for each vertex $u$ with $y'_u > 1/2$ and $y_v = 0$. In a case, $F_i = \{u, v\}$ and $y'_u = y'_v = 1/2$, we set $y_u = 1$ where $u$ is the optimal waiting location of driver $i$ with $I_i = \emptyset$. Then we assign each vertex $v \in V$ to the closest vertex $u$ in $\cup_{i = 1}^{m} F_i$ with $y_u = 1$ by setting $x_{u, v} = 1$. Note that the constructed solution $(x, y)$ satisfies the constraint of ILP~\eqref{eq:ILP_share_info}, therefore, it is a feasible solution to Problem~\ref{prob:share_info}. Also, observe that the optimal objective value to the LP relaxation is a lower-bound on the optimal value of ILP~\eqref{eq:ILP_share_info} and provides a bound on the performance of the LP-rounding algorithm. 

In the solution to the information sharing problem, if driver $i$ is selected to receive information on the location of  drivers, a snapshot of the location of drivers is presented to driver $i$ and the driver can calculate their expected profit based on complete information. This method employed at each time step and presents information to a driver if there is an opportunity to improve the expected response time.

The problem of information sharing indirectly controls the configuration of the drivers by providing information to a subset of them, however, the possible configurations are limited to the candidate waiting locations of the drivers. In the following section, we provide the details on the pay-to-control method for the service provider to optimize the global objective.

\section{Pay to control}
\label{sec:pay_to_control}
Each driver as a self-interested unit chooses its waiting location by maximizing the profit in Equation~\eqref{eq:driver_utility}. To convince the vehicles to relocate to another configuration, the service provider needs to compensate for the difference between their expected profit of the new location and their expected profit for the waiting location from Equation~\eqref{eq:driver_utility}. First, we pose the problem between the drivers and the service provider as a game. Then we provide an approximation algorithm to find the optimal policy for the service provider.
 
\subsection{Service Provider's Game}
\label{sec:service_providers-game}
Let $d_i$ be the incentive per unit distance offered to driver $i$. Let $Q = \{q_1, \ldots, q_m\}$ (resp. $Q' = \{q'_1, \ldots, q'_m\}$) be the configuration of the drivers before (resp. after) the incentive pay.  The game between the drivers and service provider consists of the following:
\begin{itemize}
\item A set of $m$ players and a service provider,
\item An action set $A_i$ for each driver $i$, which is the waiting locations in the graph, i.e. $A_i = V \ \forall i \in [m]$. The action set of the service provider is $\mathcal{Q}$; and
\item The profit function of the service provider is 
\begin{align*}
    h(Q')&= \sum_{i \in m} d_i \sigma c(q_i, q'_i) + \beta D(Q'),
\end{align*}
where $\beta \geq 0$ is a user-defined parameter that indicates the importance of the service quality for the service provider with respect to the incentive pay. For a small value of $\beta$, the incentive pay is in the priority, thus the service provider will offer the waiting locations close to the driver's desired waiting location, however, for large values of $\beta$, the service provider accepts high incentive pay to relocate the drivers to the configuration with minimum expected response time. 

\item The profit of driver $i$ is the maximum of the expected profit of the offered waiting location with incentive pay and the expected profit of the waiting location from Equation~\eqref{eq:driver_utility}, i.e.,
\begin{align*}
    \max \{(d_i &- 1) \sigma c(q_i, q'_i) +\mathcal{V}(q'_i, B_i - c(q_i, q'_i)),\\& \max_{u \in V}- \sigma c(q_i, u) + \mathcal{V}(u, B_i - c(q_i, u))\}.
\end{align*}
\end{itemize}

This is an instance of a leader-follower game~\cite{basar1999dynamic}. The service provider offers an incentive based on its utility and the drivers as followers either take the offer or reject it. The service provider is aware of the best action of the drivers given any action taken by the service provider (i.e., incentive pay and the offered waiting location). The objective is to find the optimal strategy for the service provider to minimize a linear combination of the incentive pay and the expected response time by relocating the drivers to the desired configuration.

Driver $i$ will accept the offer by the service provider to relocate to $q'_i$ only if the offered incentives surpass the best-expected profit of the driver. Since the profit functions of the drivers are known to the service provider, then the minimum $d_i$ in which the drivers will accept the offer to move to configuration $Q'$ is
\begin{align}
   \label{eq:minimum_d_i}
    d_i &= \frac{\max_{u \in V}- \sigma c(q_i, u) + \mathcal{V}_i(u, B- c(q_i, u)) }{c(q_i, q'_i)}\nonumber
    \\& - \frac{ \mathcal{V}_i(q'_i, B_i - c(q_i, q_i'))}{\sigma c(q_i, q_i')} + 1.    
\end{align}  
In the equation above, $\max_{u \in V}- \sigma c(q_i, u) + \mathcal{V}_i(u, B- c(q_i, u)) $ is the maximum expected profit of the driver $i$ by relocating to a new waiting location, and $ V_i(q'_i, B_i - c(q_i, q_i'))$ is the expected profit of driver $i$ by waiting at the location $q'_i$ offered by the service provider.  
Knowing this minimum $d_i$, the objective of the service provider becomes
\begin{align}
    \label{eq:incentive_pay_problem}
    h(Q')&= \sum_{i \in m}  \sigma c(q_i, q'_i) + \beta \sum_{u \in V} p_u \min_{i \in [m]}c(q'_i, u)
    \\& + \sum_{i \in m}\max_{u \in V}- \sigma c(q_i, u) + \mathcal{V}_i(u, B- c(q_i, u))\nonumber
    \\& - \sum_{i \in m} \mathcal{V}_i(q'_i, B_i - c(q_i, q_i')).\nonumber
\end{align}
Observe that $\sum_{i \in m}\max_{u \in V}- \sigma c(q_i, u) + \mathcal{V}_i(u, B- c(q_i, u))$ is independent of the optimization parameters. Therefore, the problem of minimizing the utility function of the service provider $h$ has the mobile facility location (MFL) problem as a special case where $V_i(v, B_i - c(u, v)) = 0$ for all $u,v \in V$ and $i\in [m]$. The MFL is a well-known NP-hard problem~\cite{ahmadian2013local} where given a metric graph $G = (F \cup D, E, c)$, mapping $\mu:D \rightarrow \mathbb{R}_+$ and a subset $Q \subseteq F \cup D$ of size $m$. The objective is to find a subset $Q' = \{q'_1, \ldots, q'_m\} \subseteq F$ minimizing 
    $\sum_{i \in [m]} c(q_i, q_i') +  \sum_{u \in D} \mu_u \min_{q' \in Q'} c(u, q').$  
\begin{remark}[Equilibrium]
The optimal solution to the problem $\min_{Q'} h(Q')$ is the equilibrium of the leader-follower game between the service provider and the drivers. Since any other configuration will increase the cost function of the service provider. In addition, By Equation~\eqref{eq:minimum_d_i}, waiting in a location other than the one suggested by the service provider will decrease driver's expected profit.  
\end{remark}
\subsection{Approximation Algorithm}
\label{subsec:approx_alg_incentive_pay}
We now propose a constant factor approximation for the minimum pay-to-control problem, namely minimizing Equation~\eqref{eq:incentive_pay_problem}. Let $w_{q'_i, I_i} = \frac{1}{\sigma} \max_{u \in V}  \sigma c(q_i, u) - \mathcal{V}_i(u, B- c(q_i, u)) + V(q'_i, B_i - c(q_i, q'_i))$, then the utility function of the service provider becomes
\begin{align*}
    h(Q') &= \sum_{i \in [m]}  \big( c(q_i, q'_i) - w_{q'_i, I_i} \big) + \beta \sum_{u \in V} p_a(u) \min_{i \in [m]}c(q'_i, u).
\end{align*}

The algorithm follows by a reduction from the minimum pay-to-control problem to MFL.

Given an instance of the minimum pay-to-control problem we construct an MFL instance as follows:
\begin{enumerate}
\item A graph $G = (Q \cup F \cup V, E, c')$ where $F$ is the set of possible waiting locations for the drivers
\item There is an edge between $q_i \in Q$ and $q' \in F$ with cost $c'(q_i, q') = c(q_i, q') - w_{q', I_i}$,
\item There is an edge between $q' \in F$ and $v \in V$ with cost $c'(q', v) = c(q', v)$.
\item The objective is to find a set of $m$ vertices in $F$ such that minimizes 
\[
        C(Q') = \sum_{i \in m}  c'(q_i, q'_i) + \beta \sum_{u \in V} p_a(u) \min_{i \in [m]}c'(q'_i, u).
\]
\end{enumerate} 
\begin{figure}
\centering
    \includegraphics[width=\linewidth]{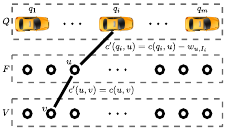}
     \caption{Constructed MFL instance for optimizing the utility of the service provider. An instance of the edges between the subsets is shown with their respective costs.}
\label{fig:incentive_pay_reduction}
\end{figure}
Figure~\ref{fig:incentive_pay_reduction} shows a schematic representation of the constructed MFL instance. 
Suppose $Q'$ is a solution to the MFL instance, we let $Q'$ be the solution of the minimum pay-to-control problem and provide the following result on the cost of the solution.
\begin{lemma}
\label{lem:pay_to_control}
Given an $\alpha$-approximation algorithm for the MFL problem, the reduction above provides an $\alpha$-approximation for the minimum pay-to-control problem.
\end{lemma}
\begin{proof}
For any $Q' \subseteq F$, by the construction of the MFL instance, we have $h(Q') = \sigma C(Q')$. Therefore, given an $\alpha$-approximation algorithm for the MFL problem, and $Q'$ obtained from the constructed MFL instance, we select $Q'$ as a solution to the minimum pay-to-control problem. Therefore, 
\[
    h(Q') = \sigma C(Q') \leq \alpha \sigma \min_{Q^* \subseteq F} C(Q^*) = \alpha \min_{Q^* \subseteq F} h(Q^*).\qedhere
\]
\end{proof}

By the result of Lemma~\ref{lem:pay_to_control}, the $3 + o(1)$-approximation algorithm for the MFL problem in~\cite{ahmadian2013local} applies to the minimum pay-to-control problem.

\section{Simulation Results}
\label{sec:simulation_results}
We evaluate the performance of the two proposed indirect controls on ride-sourcing data from Uber~\cite{UberData}. The data set consists of the pick-up time and locations from April to September 2014 in New York City, primarily Manhattan. The jammed scenario occurs frequently with high demand in an area and especially when surge price is applied for the high demand area~\cite{surgeprice, rosenblat2016algorithmic}.

To reduce the complexity of the large data set with $914$ pick-up locations, we cluster the close pick-up locations into $125$ clusters such that no two pick-up locations in a cluster are farther than $500$ meters apart. The drop-off location for each ride is selected from the same set of clusters with equal probability. Figure~\ref{fig:cluster-data} shows the clustered pick-up locations and the arrival rates for ride requests at each cluster is represented with a bar. The performance of the proposed control methods are evaluated in two scenarios: 1) the initial location of the drivers are selected uniformly randomly, and 2) jammed scenario where the drivers are initialized at $20$ closest locations to the Rockefeller center in Manhattan. 

\begin{figure}
\centering
    \includegraphics[width=\linewidth]{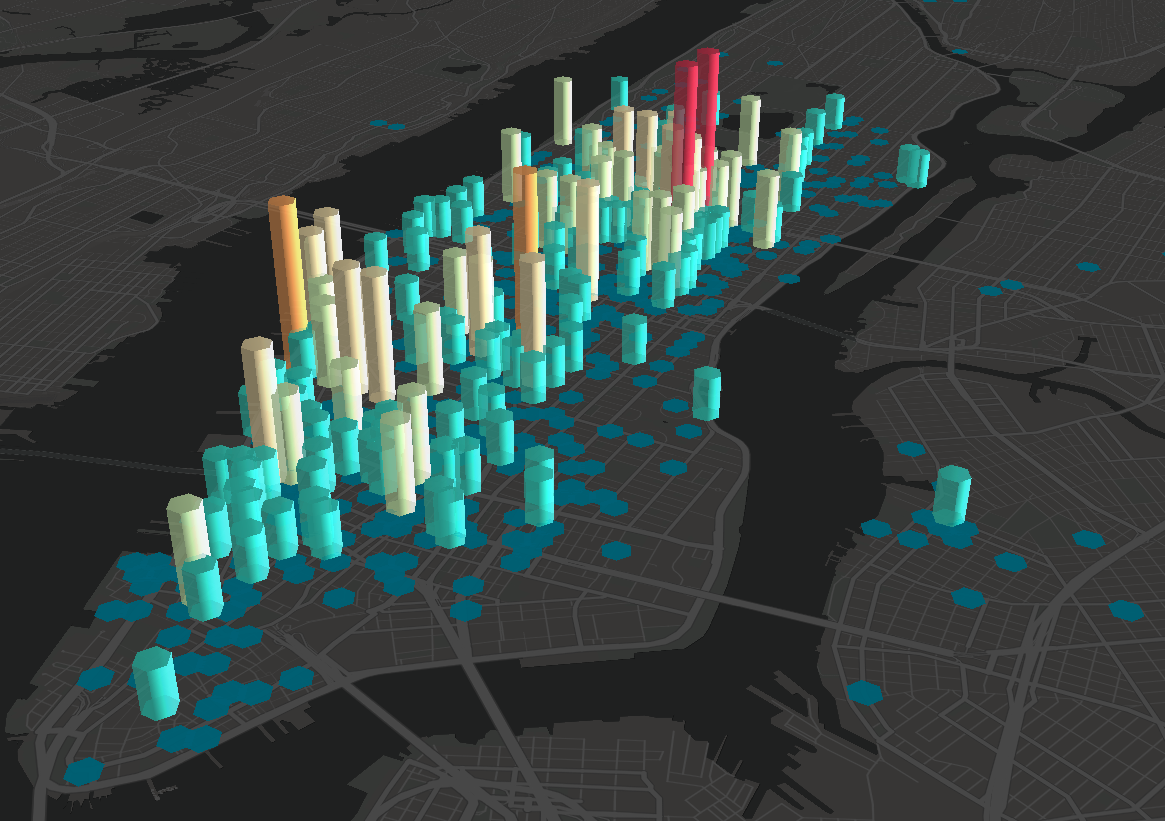}
     \caption{The set of pick-up and drop off locations in Manhattan N.Y. The bar at the location of each cluster represents the ride request arrival rate. }
\label{fig:cluster-data}
\end{figure}

Observe that the proposed algorithms to find the controls are applicable to various driver models for $\mathcal{V}$. In the following section, we propose a behavior model $\mathcal{V}$ for the drivers.

\subsection{Drivers' model}
 A driver model is a process of evaluating the expected profit of different locations at each time instance. Prior to introducing our driver model, we provide the parameters in the proposed model.   

The \emph{pick-up probability} of a location $v$ as seen by driver $i$ is the probability that a ride-request at $v$ is assigned to driver $i$ before any other ride request. Consider a configuration $Q$ of the drivers and a vertex $u$, where driver $i$ is the $k$th closest driver to the vertex location $u$. Since a request is assigned to the closest available driver, then driver $i$ can expect to be assigned to the $k$th request arriving at $u$.  Let $S_i(v, u)$ denote the number of drivers closer to $v$ than driver $i$ located at vertex $u$. Note that $S_i(v, u)$ is a function of the information of driver $i$ regarding the position of other drivers. Figure~\ref{fig:closeness} illustrates an instance with three vehicles and a pick-up location. Given the full information of the vehicle positions, $S_{\text{red}}$ at pick-up location is $2$ and $S_{\text{red}}$ is zero if  $I_{\text{red}} = \emptyset$.

The probability that a request at $v$ is assigned to driver $i$ located at $u$ as expected by driver $i$ before any other request is approximated by
\begin{align*}
    p_i(v, u) &= \Pi_{w \in V}\sum_{k = S_i(v, u) + 1}^{S_i(v, u) + S_i(w, u) + 1}{\sum_{w \in V}S_i(w, u) + 1 \choose k}\\&\big( \frac{\lambda_v}{\lambda_v + \lambda_w} \big)^{k}\big( \frac{\lambda_w}{\lambda_v + \lambda_w}\big)^{S_i(v, u) + S_i(w, u) + 1 - k}.
\end{align*}

Let $\sigma'$ be the fare per unit time of servicing a request, then the profit of a ride with pick-off at $v$ and drop-off at $w$ for driver $i$ located at $u$ is $\sigma' c(w, v) - \sigma c(u, v)$.
\begin{figure}
\centering
\begin{subfigure}{0.45\linewidth}
\centering
    \includegraphics[width=\linewidth]{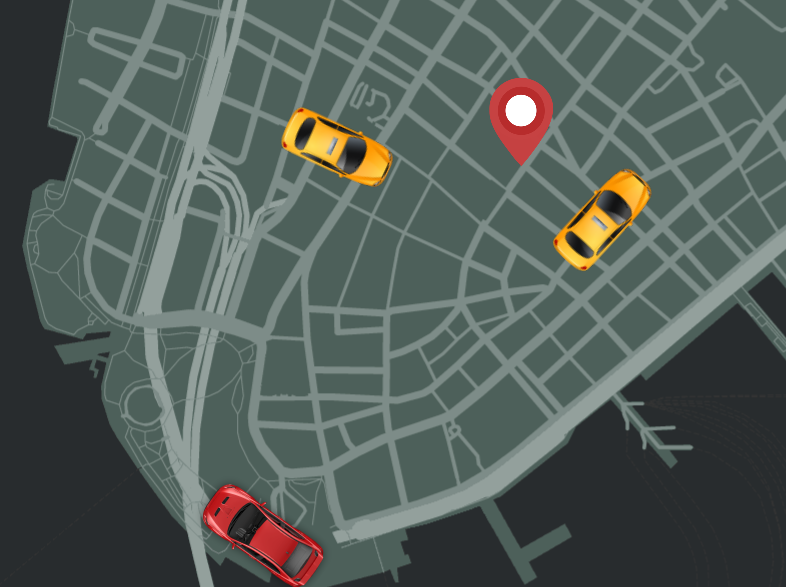}
     \caption{$I_{\mathrm{red}} = Q$}
\label{fig:closeness_1}
\end{subfigure}
\hfill
\begin{subfigure}{0.45\linewidth}
\centering
    \includegraphics[width=\linewidth]{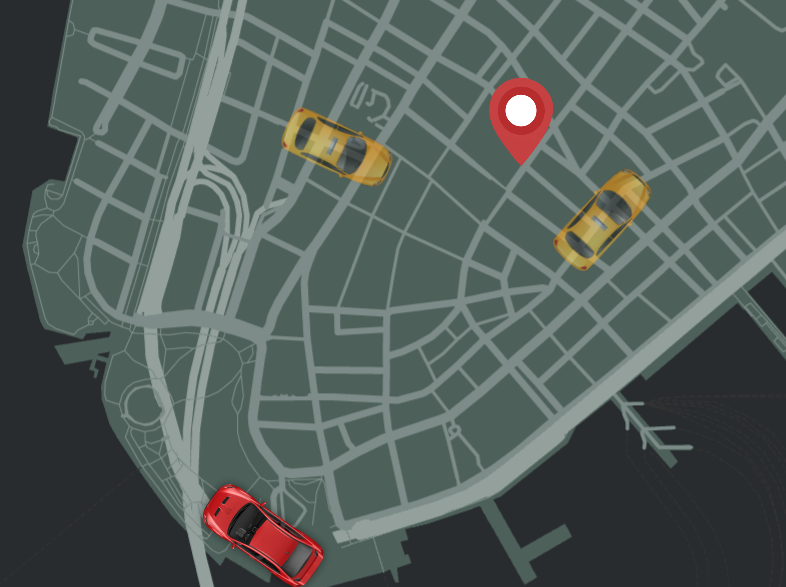}
     \caption{$I_{\mathrm{red}} = \emptyset$}
\label{fig:closeness_2}
\end{subfigure}
\caption{Environment with three drivers and a pick-up location $v$. (a) The red car has full information on the position of the other drivers, therefore, $S_{\mathrm{red}}(v, q_{\mathrm{red}}) = 2$ , (b) The positions of the other drivers are not known to the driver of red car, thus, $S_{\mathrm{red}}(v, q_{\mathrm{red}}) = 0$. }
\label{fig:closeness}
\end{figure}

 Now we define the expected profit of driver $i$ located at $u$ and working for $B_i$ period of time as follows:

\begin{align}
    \label{eq:expected_profit_driver}
    \mathcal{V}_i(u, B_i) &= \sum_{v, w \in V} p_d(w|v) \big[ p_i(v,u) \big(\max\{0, \sigma' c(w, v) \\&- \sigma c(u, v) \nonumber+ \mathcal{V}_i(w, B_i - c(u, v) - c(v, w))\}\big)\big],\nonumber
\end{align}
where $\mathcal{V}_i(u, 0) = 0$ for all $u \in V$, drivers $i \in [m]$ and $I_i \subseteq Q$. Note that,  $\mathcal{V}_i(w, B_i - c(u, v) - c(v, w))$ represents the expected profit of driver $i$ after drop-off at $w$. The conditional probability of the drop-off location is the common knowledge of the drivers and the ride-sourcing companies based on prior customer data.

Intuitively, the expected profit in Equation~\eqref{eq:expected_profit_driver} represents the total expected profit of servicing requests with pick-up location at $v$ and drop-offs at $w$ only if the 
\[
    \sigma' c(w, v)- \sigma c(u, v) + \mathcal{V}_i(w, B_i - c(u, v) - c(v, w)) \geq 0.
\]
Observe that finding $\mathcal{V}_i(u, B_i)$ for all $u \in V$ and a given $B_i$ is performed in polynomial time, however, since the drivers are not going through the calculation of $\mathcal{V}_i(u, B_i)$ for all $u \in V$, we assume that they have access to the expected profit of the vertices by experience.

Since the calculation of the expected profit for drivers for each time step is computationally expensive, we trained a Random Forest Regressor~\cite{breiman2001random} implemented by~\cite{scikit-learn} to approximate the values of $\mathcal{V}$ for each number of vehicles in the system with training data over $10000$ instances with work-day $B_i$ of $5$ average length rides, fare $\sigma' = \$ 1.06$ per mile and driving cost $\sigma = \$0.3$~\cite{uber_cost}.

\subsection{Partial Information Sharing} Figure~\ref{fig:information_sharing_res} shows the percentage improvement in the expected response time for different number of vehicles using the partial information sharing control method of Section~\ref{sec:sharing_info} in the two scenarios. Observe that as the number of drivers increases, the average improvement in the expected response time increases. However, with a large number of drivers randomly placed in the environment, the expected response time decreases and the possibility to further optimize it with information sharing is limited. On the other hand, in the jammed scenario where the drivers are concentrated in an area, the information sharing method improves the expected response time by $10\%$ on average. The information sharing method shares information on the position of the other drivers with $31.0 \%$ and $41.7 \%$ of the drivers in the random initial configuration and jammed scenarios, respectively. The results are the average of $1000$ instances for different number of drivers and scenarios. The boxes show the first, first and third quartiles of each set of experiments. The expected response time of the solution obtained from the LP-rounding algorithm of Section~\ref{sec:sharing_info} on this set of experiments is within $0.014\%$ on average of the solution of the LP relaxation of the information sharing problem. The maximum deviation from the optimal solution of the LP relaxation is $0.42\%$.

\begin{figure}
\centering
    \includegraphics[width=\linewidth]{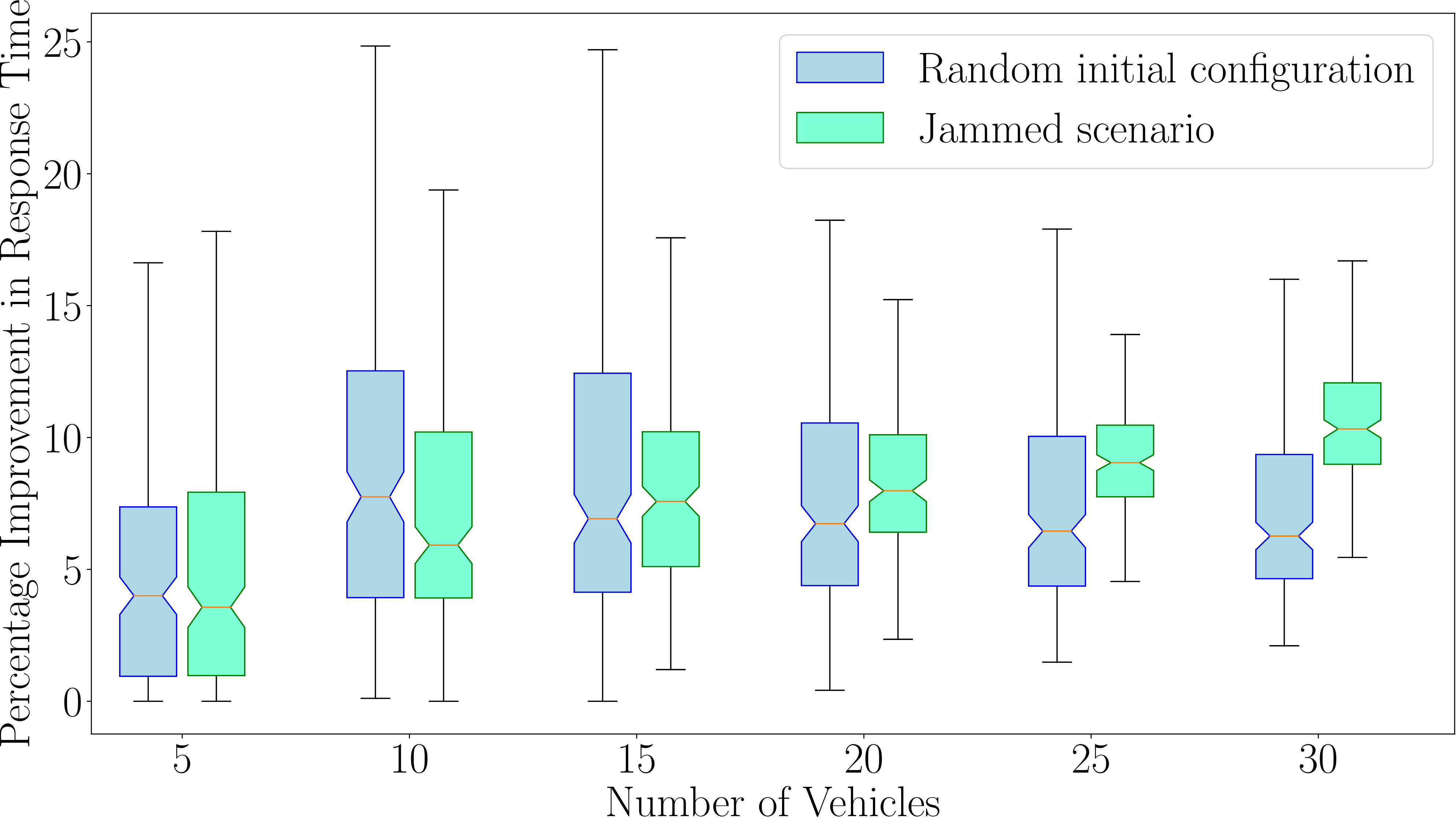}
     \caption{Improvement in expected response time by sharing location information.}
\label{fig:information_sharing_res}
\end{figure}

Figure~\ref{fig:seq_information_sharing_res} shows the expected response time of a set of $20$ drivers responding to $100$ requests arriving over time with the jammed initial configuration. In this experiment, the maximum arrival rate on the vertices is $0.03$ per minute, therefore, we assumed that there exist enough time to relocate between the ride request arrivals. Note that the partial information method maintains a low expected response time over the course of responding to $100$ requests compared to the same set of drivers with no control input on their waiting locations. The results are an average of $100$ experiments with $100$ randomly generated requests for each experiment. The lines represent the average and shaded areas represent the first and third quartiles.
\begin{figure}
\centering
    \includegraphics[width=\linewidth]{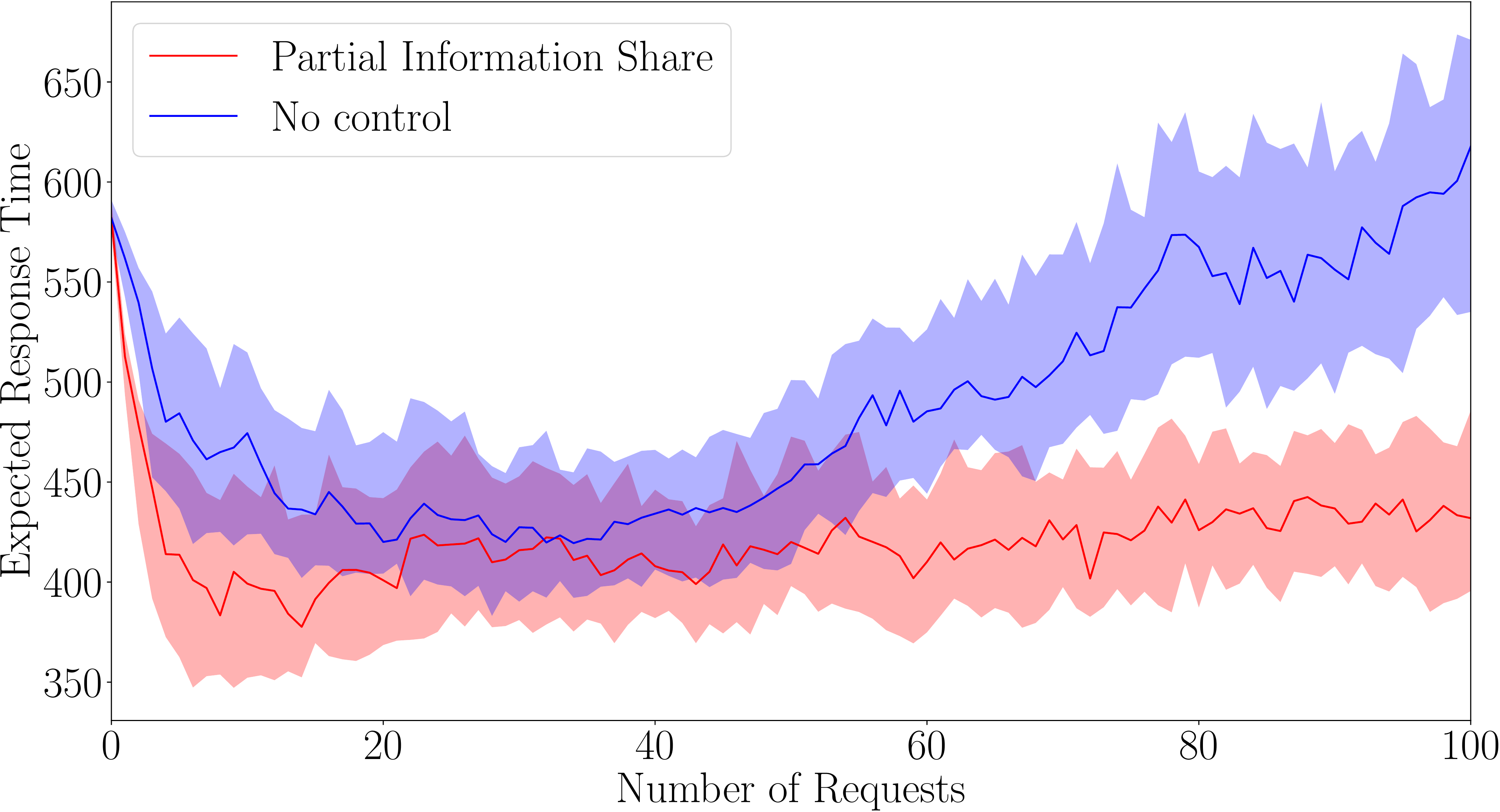}
     \caption{The expected response time of a system of $20$ drivers executing $100$ ride requests arriving over time under the partial information control method. The drivers are initialized under the jammed scenario.}
\label{fig:seq_information_sharing_res}
\end{figure}

\subsection{Pay to Control}
In this section, we evaluate the performance of the \textsc{Pay-To-Control} method for the two scenarios. Figure~\ref{fig:paid_control} illustrates the improvement in the expected response (solid lines) time and the total amount paid to the drivers (dashed lines) to relocate for different $\beta$ values in the random initial configuration scenario. The shaded area represents the first and third quartiles of $1000$ random instances for each number of vehicles. The amount paid is proportional to the average cost of riding UberXL, i.e., $\$0.3 $ per minute~\cite{uber_cost}. For larger $\beta$, the expected response time is more important than the amount paid for relocation. Therefore, with a larger number of vehicles, the \textsc{Pay-To-Control} method increases the amount paid to the drivers to minimize the expected response time. Notice that with $\beta = 10$ the expected response time has improved by $25 \%$ for $\$1$ per driver.   
\begin{figure}
\centering
    \includegraphics[width=\linewidth]{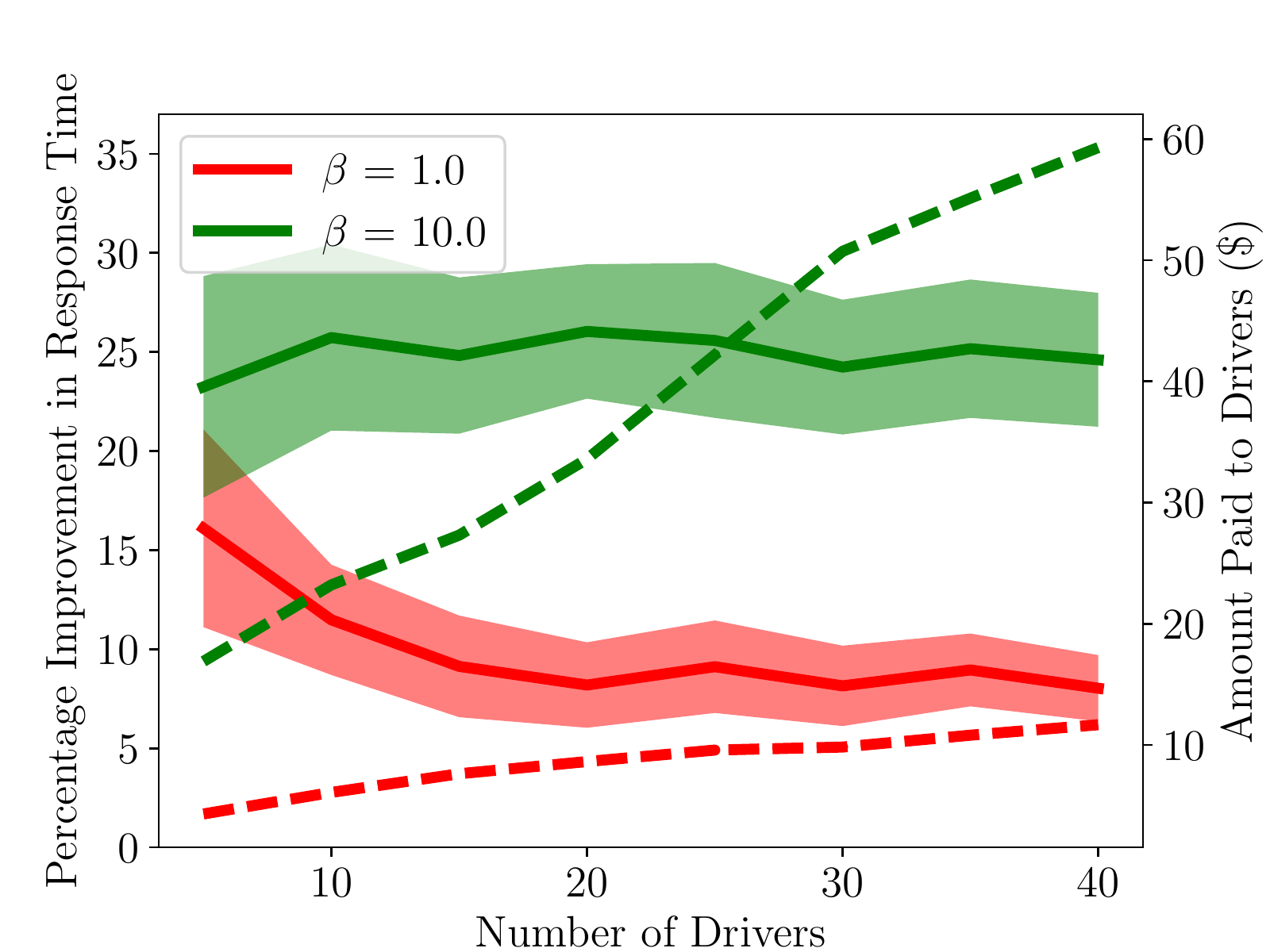}
     \caption{The percentage improvement in expected response time and the incentive pay. The solid lines represent the average improvement in the expected response time and the dashed lines represent the total amount paid to the drivers.}
\label{fig:paid_control}
\end{figure}

Next, we evaluate the performance of the pay-to-control algorithm in the jammed scenario. Figure~\ref{fig:paid_control_jam} shows that with a larger number of vehicles concentrated in a small area, the pay-to-control algorithm improves the expected response time significantly with limited amounts paid to the drivers. Notice that with $\beta = 10$ the expected response time has improved by $70 \%$ for $\$ 2$ per driver.
\begin{figure}
\centering
    \includegraphics[width=\linewidth]{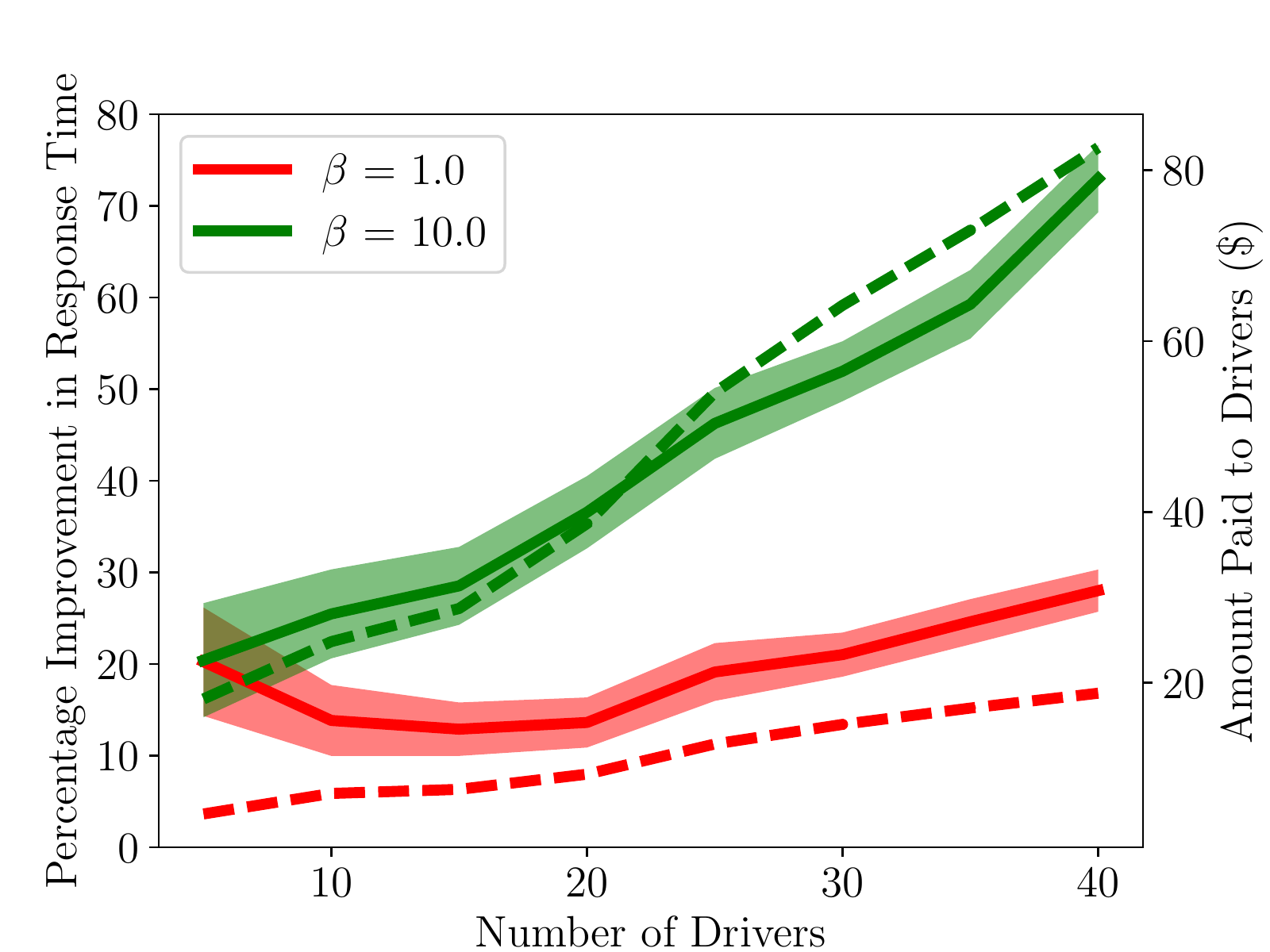}
     \caption{The percentage improvement in the expected response time and the incentive pay in the jammed scenario. The solid lines represent the average improvement in the expected response time and the dashed lines represent the total amount paid to the drivers.}
\label{fig:paid_control_jam}
\end{figure}

Figure~\ref{fig:seq_paid_control_jam} shows the expected response time and the amount paid to a set of $20$ drivers responding to $100$ requests arriving over time with the jammed initial configuration. Notice that the \textsc{Pay-To-Control} method with $\beta = 1$ maintains a low expected response time over the course of responding to $100$ requests compared to the same set of drivers with no control input on their waiting locations. The total amount paid to the drivers over the course of responding to $100$ requests is $\$ 1.87$ per request. The results are an average of $100$ experiments with $100$ randomly generated requests for each experiment. The lines represent the average and shaded areas represent the first and third quartiles.

\begin{figure}
\centering
    \includegraphics[width=\linewidth]{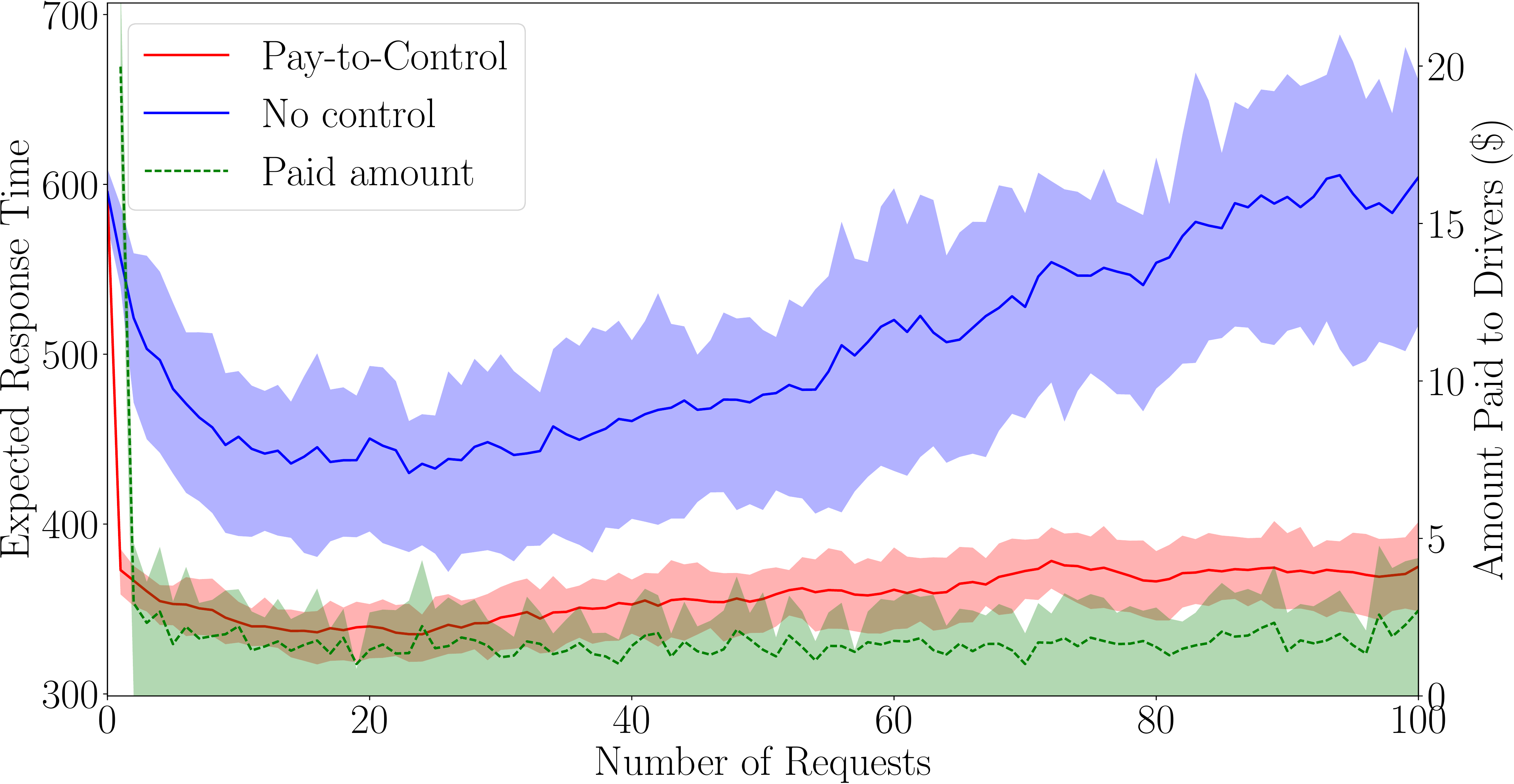}
     \caption{The expected response time and the paid amount to a system of $20$ drivers executing $100$ ride requests arriving over time under the \textsc{Pay-To-Control} method. The drivers are initialized under the jammed scenario.}
\label{fig:seq_paid_control_jam}
\end{figure}

\section{Conclusion}
\label{sec:conclusion}
This paper considered the problem of controlling self-interested drivers in ride-sourcing applications. Two indirect control methods were proposed and for each, a near-optimal algorithm was presented. The extensive results show
significant improvement in the expected response time on real-world ride-sourcing data. In addition, we hope to extend the results to capture vehicles with different capacities and ride-sharing applications.

\bibliographystyle{IEEEtran}

\begin{thebibliography}{10}
\providecommand{\url}[1]{#1}
\csname url@samestyle\endcsname
\providecommand{\newblock}{\relax}
\providecommand{\bibinfo}[2]{#2}
\providecommand{\BIBentrySTDinterwordspacing}{\spaceskip=0pt\relax}
\providecommand{\BIBentryALTinterwordstretchfactor}{4}
\providecommand{\BIBentryALTinterwordspacing}{\spaceskip=\fontdimen2\font plus
\BIBentryALTinterwordstretchfactor\fontdimen3\font minus
  \fontdimen4\font\relax}
\providecommand{\BIBforeignlanguage}[2]{{%
\expandafter\ifx\csname l@#1\endcsname\relax
\typeout{** WARNING: IEEEtran.bst: No hyphenation pattern has been}%
\typeout{** loaded for the language `#1'. Using the pattern for}%
\typeout{** the default language instead.}%
\else
\language=\csname l@#1\endcsname
\fi
#2}}
\providecommand{\BIBdecl}{\relax}
\BIBdecl

\bibitem{rayle2014app}
L.~Rayle, S.~Shaheen, N.~Chan, D.~Dai, and R.~Cervero, ``App-based, on-demand
  ride services: Comparing taxi and ridesourcing trips and user characteristics
  in san francisco university of california transportation center (uctc),''
  \emph{University of California, Berkeley, United States}, 2014.

\bibitem{surgeprice}
\BIBentryALTinterwordspacing
N.~Diakopoulos. How {U}ber surge pricing really works. [Online]. Available:
  \url{https://www.washingtonpost.com/news/wonk/wp/2015/04/17/how-uber-surge-pricing-really-works/}
\BIBentrySTDinterwordspacing

\bibitem{rosenblat2016algorithmic}
A.~Rosenblat and L.~Stark, ``Algorithmic labor and information asymmetries: A
  case study of {U}ber’s drivers,'' \emph{International Journal Of
  Communication}, 2016.

\bibitem{zhang2015performance}
W.~Zhang, S.~Guhathakurta, J.~Fang, and G.~Zhang, ``The performance and
  benefits of a shared autonomous vehicles based dynamic ridesharing system: An
  agent-based simulation approach,'' in \emph{Transportation Research Board
  94th Annual Meeting}, no. 15-2919, 2015.

\bibitem{hyland2018dynamic}
M.~Hyland and H.~S. Mahmassani, ``Dynamic autonomous vehicle fleet operations:
  Optimization-based strategies to assign {AV}s to immediate traveler demand
  requests,'' \emph{Transportation Research Part C: Emerging Technologies},
  vol.~92, pp. 278--297, 2018.

\bibitem{DBLP:journals/corr/abs-1805-02014}
M.~Chang, D.~S. Hochbaum, Q.~Spaen, and M.~Velednitsky, ``{DISPATCH:} an
  optimal algorithm for online perfect bipartite matching with i.i.d.
  arrivals,'' \emph{CoRR}, vol. abs/1805.02014, 2018.

\bibitem{maciejewski2016assignment}
M.~Maciejewski, J.~Bischoff, and K.~Nagel, ``An assignment-based approach to
  efficient real-time city-scale taxi dispatching,'' \emph{IEEE Intelligent
  Systems}, vol.~31, no.~1, pp. 68--77, 2016.

\bibitem{uberassign}
\BIBentryALTinterwordspacing
{U}ber. Driving with {U}ber, wait less, earn more. [Online]. Available:
  \url{https://www.uber.com/info/get-trips-without-waiting/}
\BIBentrySTDinterwordspacing

\bibitem{pavone2012robotic}
M.~Pavone, S.~L. Smith, E.~Frazzoli, and D.~Rus, ``Robotic load balancing for
  mobility-on-demand systems,'' \emph{The International Journal of Robotics
  Research}, vol.~31, no.~7, pp. 839--854, 2012.

\bibitem{tsao2018stochastic}
M.~Tsao, R.~Iglesias, and M.~Pavone, ``Stochastic model predictive control for
  autonomous mobility on demand,'' \emph{arXiv preprint arXiv:1804.11074},
  2018.

\bibitem{calafiore2017flow}
G.~C. Calafiore, C.~Novara, F.~Portigliotti, and A.~Rizzo, ``A flow
  optimization approach for the rebalancing of mobility on demand systems,'' in
  \emph{IEEE International Conference on Decision and Control}, 2017, pp.
  5684--5689.

\bibitem{shmoys2000approximation}
D.~B. Shmoys, ``Approximation algorithms for facility location problems,'' in
  \emph{International Workshop on Approximation Algorithms for Combinatorial
  Optimization}.\hskip 1em plus 0.5em minus 0.4em\relax Springer, 2000, pp.
  27--32.

\bibitem{arya2004local}
V.~Arya, N.~Garg, R.~Khandekar, A.~Meyerson, K.~Munagala, and V.~Pandit,
  ``Local search heuristics for k-median and facility location problems,''
  \emph{SIAM Journal on computing}, vol.~33, no.~3, pp. 544--562, 2004.

\bibitem{demaine2009minimizing}
E.~D. Demaine, M.~Hajiaghayi, H.~Mahini, A.~S. Sayedi-Roshkhar, S.~Oveisgharan,
  and M.~Zadimoghaddam, ``Minimizing movement,'' \emph{ACM Transactions on
  Algorithms (TALG)}, vol.~5, no.~3, p.~30, 2009.

\bibitem{sadeghi2018re}
A.~Sadeghi and S.~L. Smith, ``Re-deployment algorithms for multiple service
  robots to optimize task response,'' in \emph{IEEE International Conference on
  Robotics and Automation}, 2018, pp. 2356--2363.

\bibitem{bandyapadhyay2015voronoi}
S.~Bandyapadhyay, A.~Banik, S.~Das, and H.~Sarkar, ``Voronoi game on graphs,''
  \emph{Theoretical Computer Science}, vol. 562, pp. 270--282, 2015.

\bibitem{salhab2017dynamic}
R.~Salhab, J.~Le~Ny, and R.~P. Malham{\'e}, ``A dynamic ride-sourcing game with
  many drivers,'' in \emph{55th Annual Allerton Conference on Communication,
  Control, and Computing}, 2017, pp. 770--775.

\bibitem{schuler2005algorithm}
R.~Schuler, ``An algorithm for the satisfiability problem of formulas in
  conjunctive normal form,'' \emph{Journal of Algorithms}, vol.~54, no.~1, pp.
  40--44, 2005.

\bibitem{basar1999dynamic}
T.~Basar and G.~J. Olsder, \emph{Dynamic noncooperative game theory}.\hskip 1em
  plus 0.5em minus 0.4em\relax Siam, 1999, vol.~23.

\bibitem{ahmadian2013local}
S.~Ahmadian, Z.~Friggstad, and C.~Swamy, ``Local-search based approximation
  algorithms for mobile facility location problems,'' in \emph{Proceedings of
  the twenty-fourth annual ACM-SIAM symposium on Discrete algorithms}.\hskip
  1em plus 0.5em minus 0.4em\relax SIAM, 2013, pp. 1607--1621.

\bibitem{UberData}
\BIBentryALTinterwordspacing
(2014) Uber {TLC FOIL} response. [Online]. Available:
  \url{https://github.com/fivethirtyeight/uber-tlc-foil-response}
\BIBentrySTDinterwordspacing

\bibitem{breiman2001random}
L.~Breiman, ``Random forests,'' \emph{Machine learning}, vol.~45, no.~1, pp.
  5--32, 2001.

\bibitem{scikit-learn}
F.~Pedregosa, G.~Varoquaux, A.~Gramfort, V.~Michel, B.~Thirion, O.~Grisel,
  M.~Blondel, P.~Prettenhofer, R.~Weiss, V.~Dubourg, J.~Vanderplas, A.~Passos,
  D.~Cournapeau, M.~Brucher, M.~Perrot, and E.~Duchesnay, ``Scikit-learn:
  Machine learning in {P}ython,'' \emph{Journal of Machine Learning Research},
  vol.~12, pp. 2825--2830, 2011.

\bibitem{uber_cost}
\BIBentryALTinterwordspacing
{R}ide~sharing driver. {H}ow much does uber cost? {U}ber fare estimator.
  [Online]. Available: \url{https://www.ridesharingdriver.com/}
\BIBentrySTDinterwordspacing

\end{thebibliography}

\end{document}